\documentclass[11pt,rqno, a4paper]{article}

\usepackage{xcolor}
\pagecolor{white}

\usepackage{titlesec}
\titleformat*{\section}{\sc\centering\large} 
\titleformat*{\subsection}{\bf} 
\titleformat*{\subsubsection}{\it} 

\usepackage{amsmath}
\usepackage{amssymb}
\usepackage{amsfonts}
\usepackage{setspace}
\usepackage{latexsym}
\usepackage{mathrsfs}
\usepackage{epsfig}
\usepackage{amsthm}
\usepackage{color}

\usepackage[margin=2cm]{caption}

\usepackage{comment}

\newcommand{\be}{\begin{equation}}
\newcommand{\ee}{\end{equation}}

\newcommand{\vs}{\vspace{0.2cm}}

\usepackage{fancyhdr}
\newtheorem{Theorem}{Theorem}

\newtheorem{Proposition}[Theorem]{Proposition}
\newtheorem{Conjecture}[Theorem]{Conjecture}

\newtheorem{Corollary}[Theorem]{Corollary}

\newcommand{\dist}{d} 

\newcommand{\Sa}{{\rm S}^{1}} 

\usepackage{tocloft} 

\usepackage{setspace, tocloft}

\allowdisplaybreaks

\begin{document}

\thispagestyle{empty}

\begin{center}
{\Large\bf On the existence of charged electrostatic black holes in
\vs

arbitrary topology}

\vs\vs\vs

{\sc Mart\'in Reiris}

{mreiris@cmat.edu.uy}
\vs

\vspace{.4cm}

{\it Centro de Matem\'atica, Universidad de la Rep\'ublica} 

{\it Montevideo, Uruguay}
\vs\vs

\begin{abstract}
The general classification of $3+1$-static black hole solutions of the Einstein equations, with or without matter, is central in General Relativity and important in geometry. In the realm of $\Sa$-symmetric {\it vacuum} spacetimes, a recent classification proved that, without restrictions on the topology or the asymptotic behavior, black hole solutions can be only of three kinds: (i) Schwarzschild black holes, (ii) Boost black holes, or (iii) Myers-Korotkin-Nicolai black holes, each one having its distinct asymptotic and topological type. In contrast to this, very little is known about the general classification of $\Sa$-symmetric static {\it electrovacuum} black holes although examples show that, on the large picture, there should be striking differences with respect to the vacuum case. A basic question then is whether or not there are charged analogs to the static vacuum black holes of types (i), (ii) and (iii). In this article we prove the remarkable fact that, while one can `charge' the Schwarzschild solution (resulting in a Reissner-Nordstr\"om spacetime) preserving the asymptotic, one cannot do the same to the Boosts and to the Myers-Korotkin-Nicolai solutions: the addition of a small or large electric charge, if possible at all, would transform entirely their asymptotic behavior. In particular, such vacuum solutions cannot be electromagnetically perturbed. The results of this paper are consistent but go far beyond the works of Karlovini and Von Unge on periodic analogs of the Reissner-Nordstr\"om black holes.

The type of result as well as the techniques used are based on comparison geometry a la Bakry-\'Emery and appear to be entirely novel in this context. The findings point to a complex interplay between asymptotic, topology and charge in spacetime dimension $3+1$, markedly different from what occurs in higher dimensions.
\end{abstract}

\end{center}

\section{Introduction}

Since the early years of General Relativity, the classification of static solutions has been quite an important problem. Already in 2016, Schwarzschild presented his (by now) celebrated solution that has been fundamental in our understanding of black holes. The exterior, static, Schwarzschild solutions are defined on $\mathbb{R}^{3}$ minus a 3-ball, are asymptotically flat, and have a horizon at their boundary which is the characteristic feature defining black holes. Of course, the Schwarzschild solution models also the exterior gravitational field of spherically symmetric (isolated) material bodies and is thus of central importance in astrophysics. Now, since already some decades there has been increasing theoretical interest in the understanding of static black hole solutions (i.e. having a horizon) with more exotic topology, matter, and asymptotic behavior. This article investigates various aspects of the geometry and topology of electro-vacuum static and $\Sa$-symmetric 3+1 black hole solutions, which, as we will see, entangle new phenomena relative to the vacuum case.  

To frame the discussion in the following paragraphs, let us begin recalling what is known about the relation between the topology and the asymptotic behavior for static {\it vacuum} solutions, (not necessarily black hole solutions). First, the celebrated uniqueness theorem of the Schwarzschild solutions by Israel \cite{Israel}, Robinson \cite{RobinsonII} and Bunting-Masood \cite{MR876598}, asserts that a $3+1$ asymptotically flat static vacuum black hole is necessarily Schwarzschild. The fundamental hypothesis here is asymptotic flatness that ends up constraining the spatial topology to that of $\mathbb{R}^{3}$ minus a 3-ball. In particular, it is impossible to achieve static equilibrium between several black holes when the solution is vacuum and asymptotically flat. Conversely, if the spatial topology of a static vacuum black hole solution is that of Schwarzschild and the solution is metrically complete, then it is indeed asymptotically flat and therefore Schwarzschild again \cite{MR3233266}, \cite{MR3233267}. Thus, in vacuum at least, the topology and the asymptotic are tightly related. Similarly, if a static vacuum solution is asymptotically flat and without horizons (i.e. not a black hole solution) then Lichnerowicz showed that the solution is the Minkowski spacetime \cite{Lichnerowicz}. Conversely, if the spatial topology is that of the Minkowski spacetime and the vacuum solution is metrically complete, then Anderson \cite{MR1809792} proved that the solution must be the Minkowski spacetime. More in general, it was shown in \cite{PartI}, \cite{PartII} that, without topological and asymptotic assumptions, a metrically complete static vacuum solution with a non-empty horizon must be either: (i) a Schwarzschild black hole, (ii) a Boost black hole, or (iii) a black hole of Myers-Korotkin-Nicolai (MKN) type (see below). This result was refined in \cite{Reiris_2019} but assuming $\Sa$-symmetry, showing that a $\Sa$-symmetric solution of MKN-type is indeed a MKN-solution. 

The universal MKN solutions \cite{PhysRevD.35.455}, \cite{94aperiodic}, \cite{KOROTKIN1994229} are axisymmetric, static solutions, having a periodic array of infinite coaxial black holes ({\it universal} is because the spatial manifold is simply connected). In Weyl coordinates the metric is of the form,
\be
g=-e^{\omega}dt^{2}+e^{-\omega}(e^{2k}(dz^{2}+d\rho^{2})+\rho^{2}d\varphi^{2}).
\ee
The exponents $\omega$ and $k$ are independent on the axial coordinate $\varphi$ and $(z,\rho)\in \mathbb{R}_{z}\times \mathbb{R}^{+}_{\rho}$. The axis and the horizons are located over $\{\rho=0\}$ and the exponents $\omega$ and $k$ are suitably singular over the periodic array of segments representing the horizons (Weyl's coordinates degenerate on them). As $\rho\rightarrow \infty$ the metric gets asymptotically Kasner, namely,
\be
g\sim -c^{2}\rho^{\alpha}dt^{2}+a^{2}\rho^{\alpha^{2}/2-\alpha}(dz^{2}+d\rho^{2})+b^{2}\rho^{2-\alpha}d\varphi^{2}
\ee
where $0<\alpha<2$. Observe that because $\alpha<2$, the norm of the axial Killing field $\partial_{\varphi}$ tends to infinity. This information will be used later. Being periodic (the translation group on the variable $z$ is $\mathbb{Z}$), the universal MKN solutions can be quotient so that the quotient possesses only an arbitrary but finite number of black hole horizons. These are the solutions that we call MKN in this paper. The quotient manifolds are diffeomorphic to open 3-torus minus a finite number of open 3-balls (see Figure \ref{Figure21}).
\begin{figure}[h]
\centering
\includegraphics[width=6cm, height=4cm]{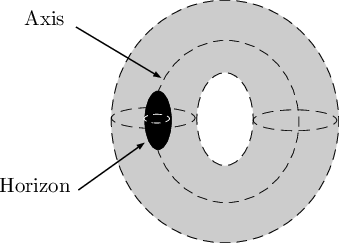}
\caption{The grey region is the spatial manifold of the MKN solutions, topologically an open solid 3-torus minus the horizons, one in this case. The only boundary of the manifold is the horizon.}
\label{Figure21}
\end{figure}

In the presence of electromagnetic matter (electric, magnetic, or both), very little is known about a general classification of static solutions. The following well-known examples show that the space of solutions is however more diverse and more complex than in vacuum, and that there is a more subtle interplay between topology and asymptotic behavior. First, asymptotically flat, charged, electrostatic black hole solutions are known to be the Reissner-Nordstr\"om solutions \cite{1988CQGra...5L.155R}. Roughly speaking, these solutions are the electrically charged Schwarzschild black holes. However, the magnetostatic Schwarzschild-Melvin solutions, which are metrically complete and have the same spatial topology as the Schwarzschild or the Reissner-Nordstr\"om solutions are asymptotic to the magnetized Melvin universe and are therefore not asymptotically flat. Roughly speaking, these solutions represent an uncharged Schwarzschild black hole embedded inside a magnetostatic universe. Similarly, the Melvin and the Bertotti-Robinson solutions are examples of magnetostatic metrically complete solutions having the same spatial topology of the Minkowski spacetime but that are not asymptotically flat. Garfinkle and Glass \cite{Garfinkle_2011} conjectured that in the presence of magnetic energy only, no other metrically complete static solution exists in that topology. 

We now wonder whether there exist charged electrostatic or magnetostatic black hole solutions in topologies like those of the MKN black holes. In this article, we prove the remarkable feature that, while the Schwarzschild black hole can be perturbed by adding a small electric charge (getting thus a $\Sa$-symmetric Reissner-Nordstr\"om black hole), the Boosts and the Myers-Korotkin-Nicolai cannot: the addition of any electric charge, small or big, keeping the $\Sa$-symmetry, would entirely affect the asymptotic behavior. We conjecture that the Boosts and the Myers-Korotkin-Nicolai topologies indeed do not accept charged electrostatic or magnetostatic black hole solutions. We note however that one can quotient the Melvin solution by a translation along the $z$-axis to obtain indeed a magnetostatic $\Sa$-symmetric and metrically complete solution whose spatial topology is that of an open solid $3$-torus. Yet such a solution is uncharged (the end of the manifold has no magnetic flux). 

The main result to be proved in this article, from which the main conclusion mentioned in the previous paragraph will follow, is a theorem about the asymptotic behavior of $\Sa$-symmetric electrostatic or magnetostatic charged solutions on spatial 3-manifolds of the form $[0,\infty)_{\rho}\times \Sa \times \Sa$, with one of the $\Sa$-factors being the orbits of the $\Sa$-symmetry (infinity is as $\rho\rightarrow \infty$). When discussing the applications a few lines below, this manifold will be thought of as just the asymptotic region of a certain solution defined on a larger manifold. The main theorem, that we state right below, says that if an $\Sa$-symmetric electrostatic or magnetostatic charged solution on such a 3-manifold has the norm of the Killing field generating the $\Sa$-symmetry bounded away from zero, then the lapse $N$ must tend to zero at infinity. 

\begin{Theorem}\label{MT} Consider an $\Sa$-symmetric charged electrostatic or magnetostatic spacetime metric in Weyl coordinates,
\be\label{FS}
{\bf g}=-e^{2\bar{u}}dt^{2}+e^{-2\bar{u}}(e^{2\bar{k}}(d\rho^{2}+dz^{2})+\rho^{2}d\varphi^{2}),
\ee
and on the $4$-manifold ${\bf M}=\mathbb{R}_{t}\times [\rho_{0},\infty)_{\rho}\times \Sa_{z}\times \Sa_{\varphi}$, ($\rho_{0}>0$). If the spatial Riemannian $3$-manifold is metrically complete and the norm of the $\Sa$-Killing field $\partial_{\varphi}$ is bounded below away from zero, then the lapse $N=e^{\bar{u}}$ tends to zero at infinity (i.e. as $\rho\rightarrow \infty$).
\end{Theorem}
This theorem implies that there are no charged electrostatic (or magnetostatic) black holes on a 3-manifold with its manifold-end (assume only one) as in the theorem and with the norm of the $S^{1}$-Killing field bounded away from zero at infinity. To see this, we argue by contradiction. Suppose that a solution like that exists. Let $\Sigma$ be its spatial 3-manifold with compact boundary (the horizon) and with its manifold-end diffeomorphic to $[0,\infty)\times \Sa\times \Sa$. As we are assuming that the solution is a black hole, we have $N|_{\partial \Sigma}=0$. On the other hand by Theorem \ref{MT} we have $N\rightarrow 0$ at infinity. Recalling finally that the lapse equation for an electrostatic solution is, 
\begin{equation}
\Delta N = |E|^{2}N,
\end{equation}
then the maximum principle tells us that $N$ must be identically zero, reaching a contradiction.

This result is consistent but goes far beyond the works of Karlovini and Von Unge on periodic analogs of the Reissner-Nordstr\"om black holes \cite{Karlovini_vonUnge}, where the periodic analogs of the Reissner-Nordstr\'om they construct have a singularity surrounding the horizon.

Now, as the norm of the $\Sa$-Killing field of the MKN black holes diverges at infinity, we conclude that it is impossible to perturb (i.e. barely deform) such vacuum black holes by adding a small charge, because such a perturbation would keep the norm of the $\Sa$-Killing field also bounded away from zero. We summarize this second striking consequence below.
\begin{Corollary}\label{COR2} There are no electrostatic $\Sa$-symmetric perturbations of the static Myers-Korotkin-Nicolai black holes with non-zero charge (no matter how small).
\end{Corollary}
More in general, we conjecture that $\Sa$-symmetric charged electrostatic solutions do not exist in configurations where the fixed set of the $\Sa$-symmetry (the `axis') is compact. In particular, we conjecture that no such type of solution exists on topologies like those of the MKN-black holes.  
\begin{Conjecture} The are no $\Sa$-symmetric black holes with compact axis and having non-zero charge at each of its ends.
\end{Conjecture}
This conjecture would follow if one could prove that the norm of the $\Sa$-Killing field of one such solution must bounded below away from zero at infinity. 

Let us say finally a few words about the arguments behind the proof of Theorem \ref{MT}. The idea is simple and is based on the observation that the lapse function of charged solutions as (\ref{FS}) but with an extra symmetry in $z$ (i.e. with metric components $z$-independent) tends to zero as $\rho$ tends to infinity. This property will be shown in section \ref{ESSS} after solving explicitly the electrostatic equations with a $\Sa_{z}\times \Sa_{\varphi}$-symmetry and inspecting the form of the lapse. Then, the rough steps to prove Theorem \ref{MT} are the following: establish first enough a priori information on the main fields and the geometry at infinity (i.e. as $\rho\rightarrow \infty$) to show that an additional asymptotic symmetry forms along the $z$-direction, and then make contact with the previous observation to conclude that the lapse tends indeed asymptotically to zero. Although it looks simple, the implementation is technically involved. One of the main problems one faces is that the structure of the electrostatic equations isn't suitable for extracting a priori estimates at infinity. This is not merely a technical problem. To circumvent this and to extract the desired estimates we will make a Bonnor-type transformation of the main fields, to obtain new ones satisfying equations with a much better structure ($f$-harmonic map equations). Yet, to extract a priori gradient estimates for the main fields one needs to work with a particular conformal metric which will be shown to be metrically complete only after using the lower positive bound of the norm of the $\Sa$-Killing field. All this will be explained in detail in the next section.

In section \ref{EEB} we recall the reduced Einstein-Maxwell equations and write down the mentioned transformation of the main fields as well as the new equations. After that, we explain how to use the lower bound of the norm of the $\Sa$-Killing field to guarantee the metric completeness of the new 2-metric $\hat{\gamma}$ to be used. The decay estimates or gradient estimates are shown in section \ref{GE}. Section \ref{ESSS} investigates electrostatic solutions with a $\Sa_{z}\times \Sa_{\varphi}$ symmetry. The proof of Theorem \ref{MT} is done in section \ref{PMT}. 

\section{The reduced electrostatic equations and a Bonnor-type transformation.} \label{EEB}

In this section, we recall the reduced electrostatic equations (see \cite{MR2003646}), explain the mentioned transformation, and explain why we need to work with a particular conformal metric. The analysis of the resulting equations and the deduction of the decay estimates is made in the next section.

The electric field $E$, which we assume has no $\partial_{\varphi}$-component (in the presence of an axis such component is always zero, see \cite{BlackHolesHeusler}), is locally the gradient of an electrostatic potential $\bar{v}=\bar{v}(z,\rho)$. Note that such potential does not exist globally in $[\rho_{0},\infty)_{\rho}\times \Sa_{\varphi}$ but it does in the universal cover where we will often work. So, in addition to the metric exponents $\bar{u}$ and $\bar{k}$ in (\ref{FS}), we have $\bar{v}$. We call $(\bar{u},\bar{v},\bar{k})$ the electrostatic data. The fields $\bar{u}$ and $\bar{v}$ satisfy the electrostatic equations,
\begin{align}
\label{EI} & \Delta_{f}\bar{u} = e^{-2\bar{u}}|\nabla \bar{v}|^{2},\\
\label{EII} & \Delta_{f}\bar{v} = 2\langle \nabla \bar{v}, \nabla \bar{u}\rangle,
\end{align}
where $f=\ln \rho$ and $\Delta_{f}\psi = \Delta \psi + \langle \nabla f, \nabla \psi\rangle$. Here, the Laplacian and inner product are taken with respect to the flat metric $d\rho^{2}+dz^{2}$ but because of the conformal invariance of the equations (\ref{EI}) and (\ref{EII}) they can be taken too with respect to the metric $\bar{\gamma}=e^{2\bar{k}}(d\rho^{2}+dz^{2})$ or to any other metric conformally related to the flat metric. Equation (\ref{EI}) is essentially the lapse equation and (\ref{EII}) is the Gauss equation of electromagnetism. The exponent $\bar{k}$ satisfies the integrability equations,
\begin{align}
\label{EIV} & \bar{k}_{\rho} = \rho((\bar{u}_{\rho}^{2}-\bar{u}_{z}^{2})-(\bar{v}_{\rho}^{2}-\bar{v}_{z}^{2})e^{-2\bar{u}}),\\
\label{EV} & \bar{k}_{z} =2\rho(\bar{u}_{\rho}\bar{u}_{z}-\bar{v}_{\rho}\bar{v}_{z}e^{-2\bar{u}}).
\end{align}
that can always be solved by line integration. 

Now, the Bonnor-type of transformation that we will use, $(\bar{u},\bar{v},\bar{k})\rightarrow (u,v,k)$ is given by  (\cite{MR2003646} pg. 524, Theorem 34.4),
\begin{align}
u & = \bar{u} - \ln \rho,\\
\nabla v & = \rho e^{-2\bar{u}} * \nabla \bar{v},\\
4k & = \bar{k}-2\bar{u}+\ln \rho.
\end{align}
Here $*$ is the Hodge-star with respect to the flat metric $d\rho^{2}+dz^{2}$. These variables satisfy,
\begin{align}
\label{SI} & \Delta_{f} u = e^{2u}|\nabla v|^{2},\\
\label{SII} & \Delta_{f} v = -2\langle \nabla u,\nabla v\rangle.
\end{align}
Again, because of the conformal invariance of the equations, the Laplacian and the inner product can be taken with respect to the flat metric $d\rho^{2}+dz^{2}$ or with respect to any metric conformally related to it. The integrability equations for $k$ are,
\begin{align}
\label{SIV} & k_{\rho}=\frac{\rho}{4}((u_{\rho}^{2}-u_{z}^{2})+(v_{\rho}^{2}-v_{z}^{2})e^{2u}),\\
\label{SV} & k_{z}=\frac{\rho}{2}(u_{\rho}u_{z}+v_{\rho}v_{z}e^{2u}).
\end{align}
Crucially, the equations (\ref{SI}) and (\ref{SII}) are now $f$-harmonic map equations into the hyperbolic plane with $2$-metric $e^{2u}(dv^{2}+du^{2})$ (for the concept of $f$-harmonic map see \cite{Chen:2012aa}). As important as this, the Ricci tensor $Ric$ of the $2$-metric $\gamma = e^{2k}(d\rho^{2}+dz^{2})$ is given by,
\be
\label{SIII} Ric=\frac{\nabla\nabla \rho}{\rho}+\frac{1}{2}(\nabla u\nabla u+e^{2u} \nabla v\nabla v).
\ee
Thus the {\it Bakry-\'Emery Ricci tensor} $Ric_{f}=Ric-\nabla\nabla f$ (called $f$-Ricci tensor from now on) takes the form,
\be
\label{BER}
Ric_{f}=\frac{1}{2}(\nabla u\nabla u+e^{2u} \nabla v\nabla v) + \nabla f\nabla f.
\ee
The $f$-harmonic map equations (\ref{SI})-(\ref{SII}) for $(u,v)$ and the expression (\ref{BER}) for the $f$-Ricci tensor of the metric $\gamma=e^{2k}(d\rho^{2}+dz^{2})$ are very structured and well suited to obtain decay estimates. In fact, (\ref{SI})-(\ref{SII})) are the Ernst equations for stationary and axisymmetric solutions of the Einstein vacuum equations. 

We will explain below the nature of such estimates and how to obtain them but there is a potential problem to notice beforehand and is that the metric $\gamma$ doesn't have to be necessarily metrically complete as $\rho\rightarrow \infty$. So far, the only information available is that the spatial part of the metric (\ref{FS}), $g = e^{-2\bar{u}}(e^{2\bar{k}}(d\rho^{2}+dz^{2})+\rho^{2}d\varphi^{2})$, is complete. This implies that the $2$-metric $\xi = e^{-2\bar{u}+2\bar{k}}(d\rho^{2}+dz^{2})$ on $[\rho_{0},\infty)_{\rho}\times \Sa_{z}$ is metrically complete but not that the metric $\gamma$ is. Even using the lower bound on the norm of the axial Killing field, $|\partial_\varphi|=\rho e^{-\bar{u}}\geq c>0$, one does not get the necessary metric completeness. But with the bound for $|\partial_{\varphi}|$ one does get the completeness of the metric $\hat{\gamma}=e^{8k}(d\rho^{2}+dz^{2})$ because,
\begin{align}
\hat{\gamma}&=e^{8k}(d\rho^{2}+dz^{2})=e^{-2\bar{u}+2k-2\bar{u}+2\ln \rho}(d\rho^{2}+dz^{2})=\\
&=|\partial_{\varphi}|^{2}e^{-2\bar{u}+2k}(d\rho^{2}+dz^{2})\geq c^{2}\xi
\end{align}
Remarkably, we will show that the $f$-Ricci tensor of the metric $\hat{\gamma}$ satisfies,
\be\label{SVI}
\hat{Ric}_{f}=2(\nabla u\nabla u +e^{2u}\nabla v\nabla v)+\nabla f\nabla f
\ee
This, together with (\ref{SI})-(\ref{SII}), form a suitable system from which estimates can be obtained. Indeed, the structure of these equations is such that one can prove, by what is by now standard technique, the decay estimates,
\be\label{EST}
|\nabla u|^{2}+e^{2u}|\nabla v|^{2}\leq \frac{c}{\hat{s}^{2}},\quad |\nabla f|^{2}\leq \frac{c}{\hat{s}^{2}}
\ee
where $c>0$ is a constant independent on the solutions, and where $\hat{s}$ is the metric $\hat{\gamma}$-distance function to the boundary $\partial S=\{\rho_{0}\}\times \Sa_{z}$ of the $2$-manifold $S:=[\rho_{0},\infty)_{\rho}\times \Sa_{z}$,
\be
\hat{s}(p)={\rm dist}_{\hat{\gamma}}(p,\partial S)
\ee
The estimates (\ref{EST}) on $u$, $v$ and $\rho$ imply directly the following important estimate for original variables $\bar{u}$ and $\bar{v}$ that we will use in the proof of the main theorem,
\be\label{ESTORIGINAL}
|\nabla \bar{u}|^{2}+e^{-2\bar{u}}|\nabla \bar{v}|^{2}\leq \frac{2c}{\hat{s}^{2}}
\ee

Observe that the Gaussian curvature of $\hat{\gamma}$ is $\hat{\kappa}=2(|\nabla u|^{2}+e^{2u}|\nabla v|^{2})$ which is non-negative. Therefore the curvature is non-negative and has scale-invariant quadratic decay. With all this type of geometry at hand, a standard asymptotic analysis will permit us to study the geometry at infinity and therefore of the fields themselves. We do that in the next section.

\section{Proof of the decay estimates} \label{GE} 
From now one we will make $S:=[\rho_{0},\infty)_{\rho}\times \Sa_{z}$.
\begin{Proposition}\label{CRUCIAL}
Let $\gamma=e^{2k}(d\rho^{2}+dz^{2})$ and $\hat{\gamma}=e^{2\hat{k}}(d\rho^{2}+dz^{2})$ with $\hat{k}=\alpha k$, and $\alpha$ a constant. Then,
\be
(\hat{\kappa}\hat{\gamma}- \frac{\hat{\nabla}\hat{\nabla}\rho}{\rho})=\alpha(\kappa\gamma-\frac{\nabla\nabla \rho}{\rho}).
\ee
\end{Proposition}

\begin{proof} We will prove that the components of $\kappa\gamma- (\nabla\nabla\rho)/\rho$ in the base $(\partial_{\rho},\partial_{z})$ are linear in $k$. The proposition then follows after replacing $k$ by $\hat{k}=\alpha k$. By the standard formula, we have,
\be
\kappa \gamma = -(\partial_{\rho}^{2}k+\partial_{z}^{2}k)(d\rho^{2}+dz^{2}),
\ee
which is linear in $k$. We prove now that $(\nabla\nabla \rho)(\partial_{z},\partial_{z})$ and $(\nabla\nabla \rho)(\partial_{\rho},\partial_{z})$ are linear in $k$ too (note that as $\rho$ is harmonic, we have $(\nabla\nabla \rho)(\partial_{\rho},\partial_{\rho})=-(\nabla\nabla \rho)(\partial_{z},\partial_{z})$). Computing we have,
\be
(\nabla\nabla \rho)(\partial_{z},\partial_{z})=-\Gamma_{z z}^{\rho}=\partial_{\rho}k,
\ee
and,
\be
(\nabla\nabla \rho)(\partial_{\rho},\partial_{z})=-\Gamma_{\rho z}^{\rho}=-\partial_{z}k,
\ee
as wished.
\end{proof}
\begin{Corollary}\label{COR1} Under the notation of Proposition \ref{CRUCIAL}. If,
\be
\kappa \gamma - \frac{\nabla\nabla \rho}{\rho} = \frac{1}{2}(\nabla u\nabla u +e^{2u}\nabla v\nabla v),
\ee
then,
\be
\hat{\kappa}\hat{\gamma}-\frac{\hat{\nabla}\hat{\nabla}\rho}{\rho} = \frac{\alpha}{2}(\nabla u\nabla u +e^{2u}\nabla v\nabla v).
\ee
In particular we get that
\be\label{BERC}
\hat{Ric}_{f}:=\hat{\kappa}\hat{\gamma} - \hat{\nabla}\hat{\nabla} f=\frac{\alpha}{2}(\nabla u\nabla u +e^{2u}\nabla v\nabla v)+\nabla f\nabla f, 
\ee
is non-negative as long as $\alpha\geq 0$.
\end{Corollary}

In the above formulae observe of course that $\nabla h = dh$ for any function, hence no need to place a hat on it. The following corollary will play a role in the proof of the main theorem.

\begin{Corollary}\label{COR2} For any $\alpha\geq 0$, the 3-metric $\hat{\gamma}+\rho^{2}d\varphi^{2}$ on $S\times \Sa_{\varphi}$ has non-negative Ricci curvature.
\end{Corollary}

\begin{proof} Consider any of the totally geodesic hypersurfaces $S\times \{\varphi\}$ and let $e_{1},e_{2}$ be an orthonormal base of tangent vectors at a point $p$ and $e_{3}$ a perpendicular unit-vector at the same point. 
By the Gauss-Codazzy we have,
\be\label{GC1}
Ric(e_{i},e_{3})=0,\ \forall i=1,2,
\ee
Also, a standard formula for warped product metrics gives,
\be\label{GC2}
Ric(e_{i},e_{j})=(\hat{\kappa}\hat{\gamma}-\frac{\nabla\nabla \rho}{\rho})(e_{i},e_{j}),\ \forall i,j=1,2.
\ee
Therefore $Ric(e_{i},e_{i})\geq 0$ by Corollary \ref{COR1}. On the other hand, we have,
\be\label{GC3}
2\hat{\kappa}=R-2Ric(e_{3},e_{3})=2\hat{\kappa}-Ric(e_{3},e_{3}),
\ee
where to obtain the second equality we used (\ref{GC2}) and the fact that $\rho$ is harmonic. We get thus $Ric(e_{3},e_{3})=0$. With these metric components, we have $Ric\geq 0$ as wished. 
\end{proof}
The following inequalities are central to obtaining the decay estimates.
\begin{Proposition}
Let $(\gamma; u,v)$ satisfy (\ref{SI})-(\ref{SII}) and (\ref{BER}), then,
\begin{align}
\label{CF1} & \Delta_{f} |\nabla f|^{2}\geq 2 |\nabla f|^{4},\\
\label{CF2} & \Delta_{f} \kappa \geq 4 \kappa^{2}.
\end{align}
\end{Proposition}

\begin{proof} Recall the Bochner formula,
\be
\frac{1}{2}\Delta |\nabla \chi|^{2}=|\nabla\nabla \chi|^{2}+\langle \nabla \chi,\nabla(\Delta_{f}\chi)\rangle + Ric_{f}(\nabla \chi,\nabla \chi),
\ee
which is valid for any function $\chi$. Make now the choice $\chi=f$. Then, as $\rho$ is harmonic we have $\Delta_{f}f=0$ and recalling (\ref{BERC}) we have $Ric_{f}(\nabla f,\nabla f)\geq |\nabla f|^{4}$. Altogether we obtain (\ref{CF1}). 

We prove now (\ref{CF2}). Now recall the general Bochner formula for a $f$-harmonic map $h:S\rightarrow X$, \cite{Chen:2012aa},
\begin{align}\label{BfHM}
\frac{1}{2}\Delta_{f}e = & |\nabla d h|^{2}+\sum_{i}\langle dh(Ric_{f}(e_{i})),dh(e_{i})\rangle-\\
& \sum_{i,j} Rm^{X}\langle dh(e_{i}),dh(e_{j}),dh(e_{i}),dh(e_{j})\rangle,
\end{align}
into a target space $X$, (for the definition see for instance \cite{Chen:2012aa}), where $e$ is the energy of the map, $Rm^{X}$ is the Riemann curvature of $X$, and $\{e_{i}\}$ is an orthonormal basis of the tangent space of $S$ at the point where it is being evaluated. We will use this formula for the map $h=(u,v):S\rightarrow X:=\mathbb{H}^{2}$, from $S$ into the hyperbolic plane $\mathbb{H}^{2}$ with metric $du^{2}+e^{2u}dv^{2}$. The energy is $e=4\kappa$. This map is $f=\ln \rho$ harmonic due to equations (\ref{SI})-(\ref{SII}). As $\mathbb{H}^{2}$ has negative Riemann curvature, the last term on the r.h.s is non-negative. The first term in the r.h.s is non-negative too and the second is greater or equal than $e^{2}/2=8\kappa^{2}$ as the following computation shows.

Indeed, we have,
\be
dh(e_{i})=du(e_{i})\partial_{u}+dv(e_{i})\partial_{v},
\ee
and,
\be
Ric_{f}(e_{i})=\sum_{j}\frac{1}{2}(du(e_{i})du(e_{j})e_{j}+e^{2u}dv(e_{i})dv(e_{j})e_{j})+df(e_{i})df(e_{j})e_{j}.
\ee
Thus, using these equations in the second term in the r.h.s of the (\ref{BfHM}) we deduce,
\begin{align}
\sum_{i}\langle dh & (Ric_{f}(e_{i})), dh(e_{i})\rangle= \\
& = |\nabla u|^{4}+\langle \nabla u,\nabla v\rangle e^{2u}+|\nabla v|^{4}e^{4u}+
\langle \nabla f,\nabla u\rangle^{2}+\langle \nabla f,\nabla v\rangle^{2}e^{2u}\\
& \geq \frac{1}{2}(|\nabla u|^{2}+e^{2u}|\nabla v|^{2})^{2}=\frac{e^{2}}{2},
\end{align}
as wished. 
\end{proof}

Performing the same computation but with $(\hat{\gamma},u,v)$ we reach the following.

\begin{Corollary} For any $\alpha\geq 0$ we have,
\begin{align}
& \hat{\Delta}_{f} \hat{\kappa} \geq 4 \hat{\kappa}^{2},\\
& \hat{\Delta}_{f} |\nabla f|^{2}\geq 2 |\nabla f|^{4},
\end{align}
where the norm $|\nabla f|$  is with respect to $\hat{\gamma}$.
\end{Corollary}

The next proposition states finally the decay estimates.

\begin{Proposition}\label{ANDESTP} There is a universal constant $c>0$ such that, for all $\alpha\geq 0$,
\be\label{ANDEST}
\hat{\kappa}\leq \frac{c^{2}}{\hat{s}^{2}},\qquad \bigg|\frac{\nabla \rho}{\rho}\bigg|^{2}\leq \frac{c^{2}}{\hat{s}^{2}},
\ee
\end{Proposition}

\begin{proof} The proof follows from Lemma 3.2 in \cite{Reiris2017}.
\end{proof}

These point-wise estimates imply higher order estimates for the derivatives of $\hat{\kappa}$ that are useful in the asymptotic analysis as we will do later. The estimates are,
\begin{equation}\label{HOE}
|\nabla^{j} \hat{\kappa}|\leq \frac{c_{j}^{2}}{\hat{s}^{2+j}}
\end{equation}

To obtain them, note that because of the scale invariance of (\ref{HOE}) it is enough to prove them at points $p$ with $\hat{s}(p)=1$. But then note that $(u',v',\rho')=(u-u(p),e^{-u(p)}v,\rho/\rho(p))$ is a new solution to (\ref{SI})-(\ref{SII}) leaving (\ref{SVI}) invariant and having $(u'(p),v'(0),\rho'(0))=(0,0,1)$. From (\ref{ANDEST}) one obtains point-wise estimates for $(u',v',\rho')$ on the ball $B_{\hat{\gamma}}(p,1/2)$ and then higher-order estimates on the universal cover of the ball (where the injectivity radius at $p$ is controlled) by standard elliptic theory on (\ref{SI})-(\ref{SII}), which imply the higher-order estimates for $\hat{\kappa}$.  

\section{$\Sa_{z}\times \Sa_{\varphi}$-symmetric electrostatic charged data} \label{ESSS}
Denote $\partial = (\partial_{z},\partial_{\rho})$. In this section we study charged data $(\bar{u},\bar{v})$ for which $\bar{u}$ and $\partial \bar{v}$ are $z$-independent. This implies that not only the metric components are $z$-independent but also the electric field is. We are going to show that if one data extends to $\rho=+\infty$ then $\bar{u}$ tends to $-\infty$ and therefore $N=e^{\bar{u}}$ tends to $0$. 

We show first that $\partial_{z}\bar{v}$ must be constant. To see this we note that if $\partial_{\rho}\bar{v}$ is $z$-independent then $\partial_{z}\partial_{\rho} \bar{v} = 0$. Interchanging derivatives and integrating in $\rho$ between $\rho_{0}$ and $\rho$, leaving $z$ fixed, we obtain $\partial_{z}\bar{v}(z,\rho)=\partial_{z}\bar{v}(z,\rho_{0})$. But as $\partial_{z}\bar{v}$ is $z$-independent we conclude that $\partial_{z}\bar{v}$ is constant, as claimed.    

Now, if $\partial_{z}\bar{v}$ is constant, then (\ref{EII}) becomes $(\rho e^{-\bar{u}}\bar{v}')'=0$ where $'=\partial_{\rho}$, or,
\be\label{INTI}
\rho e^{-2\bar{u}}\bar{v}'=\bar{Q},
\ee
where $Q=2\pi\bar{Q}\neq 0$ is the electric charge. Use (\ref{INTI}) and $\partial_{z}\bar{v}=c$ into (\ref{EII}) to get,
\be\label{TOSOL}
\bar{u}''+\frac{1}{\rho}\bar{u}'=(c^{2}e^{-2\bar{u}}+\bar{Q}^{2}\frac{e^{2\bar{u}}}{\rho^{2}})
\ee
\begin{Proposition}\label{NSETI} 
If $c\neq 0$, then no solution to (\ref{TOSOL}) extends to $\rho=\infty$.
\end{Proposition}
\begin{proof} Arguing by contradiction assume that $u(\rho)$ is a solution defined on $[\rho_{0},\infty)$. We are going to reach an impossibility. We begin performing the change of variables 
\be
\bar{u}=\hat{u}+\frac{1}{2}\ln \rho + \frac{1}{4}\ln \frac{c^{2}}{Q^{2}},
\ee
making equation (\ref{TOSOL}) into
\be\label{TOO}
\hat{u}''+\frac{1}{\rho}\hat{u}'=|\bar{Q}c|\frac{\cosh 2\hat{u}}{\rho}.
\ee
Then,
\be
(\rho\hat{u}')'=|\bar{Q}c|\cosh 2\hat{u}\geq |Qc|=:c_{1}>0.
\ee
We deduce therefore $\rho\hat{u}'\geq c_{1}\rho + c_{2}$. Thus $\rho\hat{u}'\geq c_{1}\rho/2$ if $\rho \geq \rho_{1}$ for certain $\rho_{1}$ and hence 
\be\label{UDERP}
\hat{u}'\geq \frac{c_{1}}{2}>0,
\ee
for all $\rho\geq \rho_{1}$. In particular $\hat{u}\uparrow +\infty$ as $\rho\rightarrow \infty$. Now assume $\rho\geq \rho_{2}=\max\{\rho_{1},1\}$. Then $1\geq 1/\rho$ and so we get from (\ref{TOO}),
\be
(\rho\hat{u}')'\geq |\bar{Q}c|\frac{\cosh 2\hat{u}}{\rho}.
\ee
Then, making the change of variables $x=\ln \rho$, the previous inequality reads
\be\label{BLOW}
\ddot{\hat{u}}\geq |\bar{Q}c|\cosh 2\hat{u}
\ee
where $\dot{}=d/dx$. We now work assuming $x\geq x_{2}=\ln \rho_{2}$. Multiplying (\ref{BLOW}) by $\dot{\hat{u}}$ which is positive by (\ref{UDERP}), integrating in $x$, and taking into account that $\hat{u}\rightarrow +\infty$ as $\rho\rightarrow \infty$ we deduce after some manipulations
\be 
\dot{\hat{u}}\geq c_{3}e^{2\hat{u}},
\ee
for some $c_{3}>0$ and for $x\geq x_{3}$ for some $x_{3}\geq x_{2}$. Multiplying by $e^{-2\hat{u}}$ and integrating in $x$ we deduce
\be
-e^{-\hat{u}}\geq c_{3}(x-x_{3})+e^{2\hat{u}(x_{3})}.
\ee
But as $x\rightarrow +\infty$ the left hand is negative whereas the right-hand side tends to $+\infty$. We conclude that the data cannot extend to infinity.  
\end{proof}

We now assume that $\bar{u}$ is a solution to (\ref{TOSOL}) defined on $[\rho_{0},\infty)$. We will show that $N=e^{\bar{u}}\rightarrow 0$ as $\rho\rightarrow \infty$. By the previous proposition we have $c=0$ and thus $\bar{u}$ satisfies
\be\label{TOSOLSIM}
\bar{u}''+\frac{1}{\rho}\bar{u}'=\bar{Q}^{2}\frac{e^{2\bar{u}}}{\rho^{2}}.
\ee
Then, a simple manipulation after multiplication by $\rho^{2}\bar{u}'$ gives
\be
\big((\rho \bar{u}')^{2}\big)'=\big(\bar{Q}^{2}e^{2\bar{u}}\big)',
\ee
from which we deduce
\be\label{ECE}
(\rho\bar{u}')^{2}=\bar{Q}^{2}e^{2\bar{u}}\pm a^{2},
\ee
with $a\geq 0$ a constant. Letting $\chi=e^{-\bar{u}}>0$ and defining the new variable $x= \ln \rho$ we get easily (we use $\partial_{x}\chi=\dot{\chi}$)
\be
\dot{\chi}^{2} = \bar{Q}^{2} \pm a^{2}\chi^{2},
\ee
to obtain, after $x$-derivation and cancelation of $\dot{\chi}$ on both sides
\be\label{HOLA}
\ddot{\chi}=\pm a^{2}\chi,
\ee
($\dot{\chi}=0$ or $\chi=\pm \bar{Q}/a$ is not a solution of the original problem (\ref{TOSOLSIM}) and was introduced during the deduction of (\ref{ECE})).
As a consequence we obtain the following type of solutions for $\chi$ (we take only those that are positive because $\chi>0$):
\begin{enumerate}
\item $a=0$,
\be
\chi = b+|\bar{Q}|\ln x,
\ee
\item $a>0$, case $+a^{2}$ in (\ref{HOLA}),
\be
\chi = \frac{|\bar{Q}|}{a}\sinh (\pm a x +b).
\ee
\item $a>0$, case $-a^{2}$ in (\ref{HOLA}),
\be
\chi = \frac{|\bar{Q}|}{a}\sin (\pm a x +b).
\ee
\end{enumerate}
The three expressions have to be taken only in the $x$-domain where they are positive. The only solutions extending to $x=+\infty$ tend to $+\infty$ as $x\rightarrow +\infty$. Therefore $N=\chi^{-1}$ tends to $0$ as $x\rightarrow \infty$.

\section{The proof of the main Theorem} \label{PMT}

From now on we work on $S=[\rho_{0},\infty)_{\rho}\times \Sa_{z}$ with the metric $\hat{\gamma}$ and $\alpha=4$. All estimates will be done in this Riemannian manifold. As earlier we let $\hat{s}(p)=\dist_{\hat{\gamma}}(p,\partial S)$.

We know by (\ref{ANDEST}) that $|\nabla \ln \rho|\leq c/\hat{s}$. We move first to prove that there is $b>0$ and a divergent sequence $p_{i}$ such that the opposite inequality holds, 
\be\label{GTP}
\bigg|\frac{\nabla \rho}{\rho}\bigg|\bigg|_{p_{i}}\geq \frac{b}{\hat{s}(p_{i})}.
\ee
This will turn out to be a fundamental estimate for all that follows. Let us state this in the following proposition.

\begin{Proposition}\label{SEXI} There is $b>0$ and a divergent sequence $p_{i}\in S$ such that, 
\be\label{GTP}
\bigg|\frac{\nabla \rho}{\rho}\bigg|\bigg|_{p_{i}}\geq \frac{b}{\hat{s}(p_{i})}.
\ee
\end{Proposition}

\begin{proof} The idea is to use a technique developed in \cite{PartII} for harmonic functions on 3-manifolds of non-negative Ricci curvature whose level sets are surfaces of genus greater than zero. In this instance the harmonic function will be $U=\ln \rho$ on the 3-manifold $S\times \Sa_{\varphi}$ endowed with the metric $\hat{h}=\hat{\gamma}+\rho^{2}d\varphi^{2}$ that we know has non-negative Ricci curvature. Of course, the level sets of $U$ are $2$-tori. 

Recall that $S=[\rho_{0},\infty)\times \Sa_{z}$ and that $\hat{\gamma}=e^{8k}(d\rho^{2}+dz^{2})$. Then, we can write the $\hat{\gamma}$-length-element of any circle $\{\rho\}\times \Sa_{z}$ as $d\ell = e^{4k}dz$. Therefore, taking into account that $|\nabla \rho|=e^{-4k}$ we deduce that $dz = |\nabla \rho|d\ell$, is constant. We will use this below.

Let $U_{x}^{-1}$ denote the level set of $U$: $U^{-1}_{x}=\{(p,\varphi)\in S\times \Sa_{\varphi}: U(p)=x\}$. We consider the function of $x$,
\be
G(x) = \int_{U^{-1}_{x}}|\nabla U|^{2}dA,
\ee
where $dA=\rho d\ell d\varphi$ is the area element of the level set $U^{-1}_{x}$. As $dz = |\nabla \rho|d\ell$, then we can write,
\be
G(x)=\int_{U^{-1}_{x}}|\nabla U|^{2}dA=2\pi\int_{\{\rho=e^{x}\}\times \Sa_{z}} |\nabla U| d z.
\ee
It turns out, see Proposition 4.2.14 in \cite{PartII}, that $(G'/G)'\geq 0$ from which it is deduced that $G'/G\geq (G'/G)(x_{0})=:C$. Integrating we reach,
\be\label{USI}
G(x)\geq \frac{G(x_{0})}{e^{Cx_{0}}}e^{Cx}. 
\ee
If $C>0$ then $G(x)\geq G(x_{0})>0$ for all $x\geq x_{0}$. But, as $dz$ is constant, then, for every $x$ there is a point $p\in \{\rho=e^{x}\}\times \Sa_{z}$ such that $G(x)=(2\pi)^{2}(|\nabla \rho|/\rho)(p)\geq G(x_{0})>0$. Therefore $|\nabla \rho|/\rho$ does not go to zero at infinity which is impossible. Thus $C<0$. Then, using (\ref{USI}) we have,
\be
\int_{\{\rho=e^{x}\}\times \Sa_{z}} \rho^{-C}\bigg|\frac{\nabla \rho}{\rho}\bigg|dz\geq d_{1},
\ee
for some $d_{1}>0$. Therefore, for every $x$ there is $p_{x}\in \{\rho=e^{x}\}\times \Sa_{z}$ such that,
\be\label{UTIL}
\big(\rho^{-C}\bigg|\frac{\nabla \rho}{\rho}\bigg|\big)\bigg|_{p_{x}}\geq d_{2},
\ee
for $d_{2}=d_{1}/2\pi>0$. Now, assume that for every $\epsilon=-1/(2C)$ there is $\hat{s}_{\epsilon}$ such that,
\be
\bigg|\frac{\nabla \rho}{\rho}\bigg|\bigg|_{q}\leq \frac{\epsilon}{\hat{s}(q)}
\ee
for all $q$ such that $\hat{s}(q)\geq \hat{s}_{\epsilon}$. Then, a simple integration gives $\rho(q) \leq d_{3}\hat{s}^{\epsilon}(q)$ and for some $d_{3}>0$. Then,
\be
d_{4}\hat{s}(q)^{-C\epsilon-1}\geq \big(\rho^{-C}\bigg|\frac{\nabla \rho}{\rho}\bigg|\big)\bigg|_{q}
\ee
for some $d_{4}$ and for all $q$ such that $\hat{s}(q)\geq \hat{s}_{\epsilon}(q)$. If $q=p_{x}$ then using (\ref{UTIL}) we get,
\be
\hat{s}(p_{x})^{-1/2}\geq d_{5}>0.
\ee
for some $d_{5}>0$. But then the left hand side tends to zero as $x\rightarrow \infty$ which is impossible. We conclude then that there is a divergent sequence of points $p_{i}$ such that,
\be
\bigg|\frac{\nabla \rho}{\rho}\bigg|\bigg|_{p_{i}}\geq \frac{\epsilon}{\hat{s}(p_{i})},
\ee
where $\epsilon = \frac{-1}{2C}$. It follows (\ref{GTP}).
\end{proof}

Let us explore now some consequences of what we just proved. To start lets prove that the area growth on $(S;\hat{\gamma})$ is sub-quadratic. Let $B(r)=\{p\in S: \dist_{\hat{\gamma}}(p,\partial S)< r\}$ and let $A(B(r))$ be the $\hat{\gamma}$-area of $B(r)$. As $\hat{\kappa}\geq 0$, then the Bishop-Gromov monotonicity implies that the quotient $A(B(r))/r^{2}$ is monotonically non-increasing as $r$ increases. Thus $A(B(r))/r^{2}\rightarrow \mu\geq 0$ as $r\rightarrow \infty$. We prove now that indeed $\mu=0$, namely, that the area growth is sub-quadratic. 
\begin{Proposition}\label{SQVG} The manifold $(S;\hat{\gamma})$ has sub-quadratic area growth, that is $\mu=0$.
\end{Proposition}
\begin{proof}
Assume that $\mu>0$. Then by standard arguments using Bishop-Gromov's monotonicity and the scale invariant bound (\ref{HOE}) on the Gaussian curvature $\hat{\kappa}$ and its derivatives, imply that the manifold $(S;\hat{\gamma})$ has a flat cone at infinity. Let us recall this. Consider any divergent sequence $p_{i}$ and let $r_{i}=\hat{s}(p_{i})$. Given $0<c<d$ let $\mathcal{A}_{r_{i}}(c,d)=\{p:c\leq \hat{s}_{r_{i}}(p)\leq d\}$. Then, there is a divergent sequence $k_{i}\rightarrow \infty$ such that the pointed annuli $(\mathcal{A}_{r_{i}}(2^{-k_{i}},2^{k_{i}}),p_{i})$ endowed with the scaled metric $\hat{\gamma}_{r_{i}}=\hat{\gamma}/r_{i}^{2}$ converge in $C^{\infty}$ and in the pointed sense (see \cite{MR2243772}) to the flat cone $((0,\infty)_{r}\times \Sa_{\phi};\hat{\gamma}_{\infty}=dr^{2}+(\mu^{2}r^{2}/\pi^{2})d\phi^{2},p_{\infty}=(r,0)$). Choose now $p_{i}$ to be the sequence in Proposition \ref{SEXI}. Define the harmonic function $\hat{\rho}_{i}=\rho/\rho(p_{i})$. Then, as $\hat{\rho}(p_{i})=1$ and $|\nabla \hat{\rho}_{i}|_{r_{i}}(p_{i})\geq b$ we can use the scale invariant bound (\ref{ANDEST}) and standard elliptic estimates to conclude that $\hat{\rho}_{i}$ sub-converges to a limit non-negative and non-constant function $\hat{\rho}_{\infty}$ on the limit cone. But as $\hat{\kappa}_{\infty}=0$ we conclude from (\ref{SVI}), (that passes to the limit), that $\hat{\nabla}\hat{\nabla} \hat{\rho}_{\infty}=0$. Take then any infinite geodesic $\xi(t)$ on the limit cone such that $\xi(0)=p_{\infty}$ and such that,
\be\label{VN}
\frac{d}{dt}(\hat{\rho}_{\infty}(\xi(t)))\big|_{t=0}=\langle \hat{\nabla} \hat{\rho}_{\infty},\xi'(0)\rangle<0,
\ee
(to see that $\xi(t)$ always exists note that any inextensible geodesic except the radial ones diverge to infinity in both directions). But then, as $\hat{\nabla}\hat{\nabla} \hat{\rho}_{\infty}=0$ we deduce that $d^{2} \hat{\rho}_{\infty}(\xi(t))/dt^{2}=0$ and therefore that $\hat{\rho}(\xi(t))=(\frac{d}{dt}(\hat{\rho}_{\infty}(\xi(t)))\big|_{t=0})t+\hat{\rho}(p_{\infty})$. But because $\frac{d}{dt}(\hat{\rho}_{\infty}(\xi(t)))\big|_{t=0}<0$ we have $\hat{\rho}_{\infty}(t)<0$ for some $t>0$ which is not possible.
\end{proof}

We move now to prove Theorem \ref{MT}.
\begin{proof}[Proof of Theorem \ref{MT}] For $x\geq \rho_{0}$ let $\rho^{-1}_{x}$ be the level set of $\rho$, $\rho^{-1}_{x}=\{\rho=x\}\times \Sa_{z}$. We will show that there is a sequence of level sets $\rho^{-1}_{x_{i}}$ with $x_{i}\uparrow +\infty$, ($x_{i}$ is strictly increasing), on which $\bar{u}$ tends to $-\infty$, that is,
\be
\max\{\bar{u}(p):p\in \rho^{-1}_{x_{i}}\}\rightarrow -\infty.
\ee
Then, as $\hat{\Delta} \bar{u}\geq 0$, it follows by the maximum principle that,
\begin{align}
\max\{\bar{u}(p):x_{i-1} & \leq \rho(p) \leq x_{i+1}\} \leq \\
& \leq \max\{\max\{\bar{u}(p):p\in \rho^{-1}_{x_{i-1}}\},\max\{\bar{u}(p):p\in \rho^{-1}_{x_{i}}\}\}\rightarrow -\infty
\end{align}
It follows that $N=e^{\bar{u}}$ tends uniformly to zero at infinity. 

Let $p_{i}$ be the sequence of Proposition \ref{SEXI}. Again, let $r_{i}=\hat{s}(p_{i})$ and let $\hat{\gamma}_{r_{i}}=\hat{\gamma}/r^{2}_{i}$. Scaled quantities constructed out of the metric $\hat{\gamma}_{r_{i}}$ are denoted with a subindex $r_{i}$. For instance the Gaussian curvature of $\hat{\gamma}_{r_{i}}$ is $\hat{\kappa}_{r_{i}}:=\hat{\kappa}/r_{i}^{2}$ where $\hat{\kappa}$ is the Gaussian curvature of $\hat{\gamma}$. We will scale too the harmonic function $\rho$ as $\rho_{r_{i}}:=\rho/\rho(p_{i})$ and note that we have, 
\be\label{EQA}
\rho_{r_{i}}(p_{i})=1,
\ee
also,
\be\label{EQB}
\left|\frac{\nabla \rho_{r_{i}}}{\rho_{r_{i}}}\right|^{2}_{r_{i}}\bigg|_{p_{i}}\geq b>0,
\ee
and
\be\label{EQC}
\left|\frac{\nabla \rho_{r_{i}}}{\rho_{r_{i}}}\right|^{2}_{r_{i}}\bigg|_{p}\leq \frac{c^{2}}{\hat{s}_{r_{i}}^{2}(p)},
\ee
valid on all $p$.

For given integer $k\geq 2$ the sequence of Riemannian annuli $(\mathcal{A}_{\hat{r}_{i}}(1/2,2^{k}),\hat{\gamma}_{r_{i}})$ collapses in volume with bounded curvature by Proposition \ref{SQVG}. It follows from standard theory (see \cite{MR950552}) and standard elliptic estimates for the system (\ref{SI})-(\ref{SII})-(\ref{SVI}) that there is a sequence of integers $k_{i}\rightarrow \infty$, a sequence of neighborhoods $\mathcal{U}_{i}$ of $\mathcal{A}_{\hat{r}_{i}}(1/2,2^{k})$ and of finite coverings (unwrapping) $\pi_{i}:\tilde{\mathcal{U}}_{i}\rightarrow \mathcal{U}_{i}$, such that $(\tilde{\mathcal{U}}_{i};\pi^{*}_{i}\hat{\gamma}_{r_{i}})$ sub-converges in $C^{\infty}$ to a $\Sa$-symmetric surface $(\tilde{\mathcal{U}},\tilde{\gamma})$ (half an infinite cylinder) of non-negative Gaussian curvature, with one end and one boundary component. Let $\tilde{s}$ be the limit of the lifted distance functions $\hat{s}_{r_{i}}$. We can write $\tilde{\mathcal{U}}=[1/2,\infty)_{\tilde{s}}\times \Sa$ because the distance function $\tilde{s}$ is invariant under the $\Sa$-action. Observe that because of (\ref{EQA}), (\ref{EQB}) and (\ref{EQC}) the functions $\rho_{r_{i}}\circ \pi_{i}$ sub-converge to a $\Sa$-symmetric non-constant $\Sa$-invariant positive harmonic function $\tilde{\rho}$ on $\tilde{\mathcal{U}}$. We will show now that $\tilde{\rho}$ tends to infinity at infinity. Following a similar notation as earlier, for $x\geq 1/2$ let $\tilde{s}^{-1}_{x}:=\{\tilde{s}=x\}\times \Sa$, that is, the level set of $\tilde{s}$ at the value $x$.  Now, observe that because $\tilde{\rho}$ is a $\tilde{\gamma}$-harmonic non-constant and $\Sa$-symmetric function, and because of (\ref{EQB}) we have,
\be\label{EQD}
\tilde{\ell}(x)\partial_{x}\tilde{\rho}(x)=b_{1}\neq 0,
\ee
for all $x\geq 1/2$, where $\tilde{\ell}(x)$ is the $\tilde{\gamma}$-length of $\tilde{s}^{-1}_{x}$. But because the Gaussian curvature $\tilde{\kappa}$ of $\tilde{\gamma}$ is non-negative it is easy to see that $\tilde{\ell}(x)$ grows at most linearly, namely, $\tilde{\ell}(x)\leq b_{2}x$ with $b_{2}>0$. This information together with (\ref{EQD}) implies that $\tilde{\rho}$ tends to infinity as $\tilde{s}$ tends to infinity. The importance of this is that we can use $\tilde{\rho}$ as a harmonic coordinate extending to infinity. We will use this fact at the end of the proof.

The task now is to analyze the potentials $\bar{u}$ and $\bar{v}$ over $\mathcal{U}_{i}$ and by doing so prove that a subsequence of $\bar{u}$ seen over $\mathcal{U}_{i}$ tends to $-\infty$, and therefore that $N$ tends to $0$ by the argument at the beginning of the proof. We will prove this by contradiction, assuming that the maximum of $\bar{u}$ over the $\mathcal{U}_{i}$ does not tend to $-\infty$ on any subsequence. We will be able to pass to a limit and show that this would contradict Proposition \ref{NSETI} which says that charged $\Sa$-symmetric potentials do not extend to infinity.  

Assume that the maximum of $\bar{u}$ over the $\mathcal{U}_{i}$ does not tend to $-\infty$ on any subsequence. To start observe that, before passing to any limit, we have, 
\be
\int_{\rho^{-1}_{x}}\nabla_{n}\rho d\ell = c_{1}>0,
\ee
for any $x$, where $n$ here is the outgoing normal to the loop $\rho^{-1}_{x}$. Recall that $\rho_{i}=\rho(p_{i})$. Now, from the previous identity, we deduce that there is $q_{i}\in \rho^{-1}_{x=\rho_{i}}$ such that,
\be
\rho_{i}\ell_{i}\left|\frac{\nabla_{n}\rho}{\rho}\right|\bigg|_{q_{i}} = c_{1}
\ee
Scaling $n_{r_{i}}=r_{i}n$ and $\ell_{r_{i}}=\ell_{i}/_{r_{i}}$ preserves the previous identity, namely,
 \be
\rho_{i}\ell_{r_{i}}\left|\frac{\nabla_{n_{r_{i}}}\rho}{\rho}\right|\bigg|_{q_{i}} = c_{1}
\ee
As,
\be
\left|\frac{\nabla_{n_{r_{i}}}\rho}{\rho}\right|\bigg|_{q_{i}}\rightarrow b_{2}>0
\ee
we conclude that,
\be
\rho_{i}\ell_{r_{i}}\rightarrow c_{2} >0.
\ee
We will use this information below.

Now, the charge $Q$ has the expression,
\be
Q=\int_{\rho^{-1}_{x=\rho_{i}}} (\rho e^{-2\bar{u}}\nabla_{n}\bar{v})d\ell
\ee
So, again, there are $q_{i}\in \rho^{-1}_{x=\rho_{i}}$ such that,
\be
Q = \rho_{i}\ell_{i} (e^{-2\bar{u}}\nabla_{n}\bar{v})\bigg|_{q_{i}}. 
\ee
It is more convenient to write this as,
\be\label{QFUND}
Q = \underbrace{(\rho_{i}\ell_{r_{i}})\vphantom{\nabla_{n_{r_{i}}} }}_{\rm (I)}\underbrace{(e^{-\bar{u}})|_{q_{i}}  \vphantom{\nabla_{n_{r_{i}}}} }_{\rm (II)}\underbrace{(e^{-\bar{u}}\nabla_{n_{r_{i}}}\bar{v})|_{q_{i}}}_{\rm (III)}. 
\ee
We will think the terms (I), (II) and (III) which depend on $i$ as sequences (I)$_{i}$, (II)$_{i}$ and (III)$_{i}$. We know that (I)$_{i}$ tends to a non-zero constant as $i\rightarrow \infty$. So, if a subsequence (III)$_{i_{j}}$ of (III)$_{i}$ tends to zero then (II)$_{i_{j}}$ tends to infinity because $Q\neq 0$. Thus we get a divergent sequence of points $q_{i_{j}}$ on which $\bar{u}(q_{i_{j}})$ tends to $-\infty$ which is not the case by the initial assumption. On the other hand as (III)$_{i}$ is bounded above by (\ref{ESTORIGINAL}) we deduce that there are $0<m_{1}<m_{2}$ such that,
\be\label{LAST}
m_{1}\leq {\rm (III)}_{i}\leq m_{2}
\ee
When we consider the uniform bounds for $|\tilde{u}(q_{i_{j}})|$, the estimate (\ref{LAST}) and the estimates (\ref{ESTORIGINAL}), we deduce that one can extract a subsequence of the variables $\bar{u}$ and $\bar{v}-\bar{v}(q_{i_{j}})$ lifted to $\tilde{U}_{i_{j}}$, converging in $C^{\infty}$ to a {\it charged} $\Sa$-symmetric solution of (\ref{EI})-(\ref{EII}) over $\tilde{U}=[\tilde{\rho}_{0},\infty)_{\tilde{\rho}}\times \Sa$. But we know from section \ref{ESSS} that such a solution does not exist and we reach thus a contradiction.  
\end{proof}

\noindent {\bf Acknowledgement}. I would like to thank Marcus Khuri and Gilbert Weinstein for having shown me (before this work started) that charged solutions having a suitable expansion at infinity do not exist. 

\bibliographystyle{plain}
\bibliography{Master}

\begin{thebibliography}{10}

\bibitem{MR1809792}
Michael~T. Anderson.
\newblock On the structure of solutions to the static vacuum {E}instein
  equations.
\newblock {\em Ann. Henri Poincar\'e}, 1(6):995--1042, 2000.

\bibitem{MR876598}
Gary~L. Bunting and A.~K.~M. Masood-ul Alam.
\newblock Nonexistence of multiple black holes in asymptotically {E}uclidean
  static vacuum space-time.
\newblock {\em Gen. Relativity Gravitation}, 19(2):147--154, 1987.

\bibitem{Chen:2012aa}
Qun Chen, J{\"u}rgen Jost, and Hongbing Qiu.
\newblock Existence and liouville theorems for v -harmonic maps from complete
  manifolds.
\newblock {\em Annals of Global Analysis and Geometry}, 42(4):565--584, 2012.

\bibitem{MR950552}
Kenji Fukaya.
\newblock A boundary of the set of the {R}iemannian manifolds with bounded
  curvatures and diameters.
\newblock {\em J. Differential Geom.}, 28(1):1--21, 1988.

\bibitem{Garfinkle_2011}
D~Garfinkle and E~N Glass.
\newblock Bertotti--robinson and melvin spacetimes.
\newblock {\em Classical and Quantum Gravity}, 28(21):215012, sep 2011.

\bibitem{BlackHolesHeusler}
Markus Heusler.
\newblock {\em Black Hole Uniqueness Theorems}.
\newblock Cambridge Lecture Notes in Physics. Cambridge University Press, 1996.

\bibitem{Israel}
Israel.
\newblock Event horizons in static vacuum space-times.
\newblock {\em Phys. Review}, vol.164:5:1776--1779, 1967.

\bibitem{Karlovini_vonUnge}
Max Karlovini and Rikard von Unge.
\newblock Charged black holes in compactified spacetimes.
\newblock {\em Physical Review D}, 72, 06 2005.

\bibitem{KOROTKIN1994229}
D.~Korotkin and H.~Nicolai.
\newblock The {E}rnst equation on a riemann surface.
\newblock {\em Nuclear Physics B}, 429(1):229 -- 254, 1994.

\bibitem{94aperiodic}
D.~Korotkin and H.~Nicolai.
\newblock A periodic analog of the schwarzschild solution.
  arxiv:gr-qc/9403029v1, 1994.

\bibitem{Lichnerowicz}
Andr\'e Lichnerowicz.
\newblock Theories relativiste de la gravitation et de l electromagnetisme.
\newblock {\em Masson, Paris 1955}.

\bibitem{PhysRevD.35.455}
R.~C. Myers.
\newblock Higher-dimensional black holes in compactified space-times.
\newblock {\em Phys. Rev. D}, 35:455--466, Jan 1987.

\bibitem{MR2243772}
Peter Petersen.
\newblock {\em Riemannian geometry}, volume 171 of {\em Graduate Texts in
  Mathematics}.
\newblock Springer, New York, second edition, 2006.

\bibitem{PartI}
Martin Reiris.
\newblock A classification theorem for static vacuum black-holes, part {I}: the
  study of the lapse. {A}r{X}iv 1806.00819, 2018. {T}o appear in: {P}ure and
  applied mathematics quarterly.

\bibitem{PartII}
Martin Reiris.
\newblock A classification theorem for static vacuum black-holes, part {II}:
  the study of the asymptotic. {A}r{X}iv 1507.04570, 2018. {T}o appear in:
  {P}ure and applied mathematics quarterly.

\bibitem{MR3233266}
Martin Reiris.
\newblock Stationary solutions and asymptotic flatness {I}.
\newblock {\em Classical Quantum Gravity}, 31(15):155012, 33, 2014.

\bibitem{MR3233267}
Martin Reiris.
\newblock Stationary solutions and asymptotic flatness {II}.
\newblock {\em Classical Quantum Gravity}, 31(15):155013, 18, 2014.

\bibitem{Reiris2017}
Mart{\'\i}n Reiris.
\newblock On static solutions of the einstein-scalar field equations.
\newblock {\em General Relativity and Gravitation}, 49(3):46, 2017.

\bibitem{Reiris_2019}
Mart{\'\i}n Reiris and Javier Peraza.
\newblock A complete classification of s1-symmetric static vacuum black holes.
\newblock {\em Classical and Quantum Gravity}, 36(22):225012, oct 2019.

\bibitem{RobinsonII}
D.C. Robinson.
\newblock A simple proof of the generalization of israel's theorem.
\newblock {\em General Relativity and Gravitation}, 8(8):695--698, 1977.

\bibitem{1988CQGra...5L.155R}
Peter {Ruback}.
\newblock {A new uniqueness theorem for charged black holes}.
\newblock {\em Classical and Quantum Gravity}, 5(10):L155--L159, October 1988.

\bibitem{MR2003646}
Hans Stephani, Dietrich Kramer, Malcolm MacCallum, Cornelius Hoenselaers, and
  Eduard Herlt.
\newblock {\em Exact solutions of {E}instein's field equations}.
\newblock Cambridge Monographs on Mathematical Physics. Cambridge University
  Press, Cambridge, second edition, 2003.

\end{thebibliography}

\end{document}